\documentclass[11pt]{amsart}
\usepackage[margin=1in]{geometry}
\usepackage{amssymb}
\usepackage{mathtools}
\usepackage[thinc]{esdiff}
\usepackage{amsaddr}
\usepackage[hidelinks]{hyperref}
\usepackage{cleveref}
\usepackage{graphicx}
\usepackage{needspace}
\usepackage{natbib}
\usepackage{ifpdf}

\allowdisplaybreaks[1]

\DeclareMathOperator*{\pr}{pr}
\DeclareMathOperator*{\E}{\textit{E}}

\theoremstyle{plain}
\newtheorem{theorem}{Theorem}
\newtheorem{lemma}{Lemma}
\newtheorem{proposition}{Proposition}

\theoremstyle{definition}
\newtheorem{example}{Example}
\newtheorem{algo}{Algorithm}

\crefname{crefAppendix}{appendix}{appendices}
\Crefname{crefAppendix}{Appendix}{Appendices}

\title{Sampling on Discrete Spaces with Temporal Point Processes}
\author[C. A. Stewart and M. Sahani]{Cameron A. Stewart and Maneesh Sahani}
\address{Gatsby Computational Neuroscience Unit, University College London,\\25 Howland Street, London W1T 4JG, U.K.}
\email{cameron@cameronastew.art}
\email{maneesh@gatsby.ucl.ac.uk}
\thanks{This work was supported by the Gatsby Charitable Foundation (grant numbers GAT 3850 and GAT 4058).}


\begin{document}

\begin{abstract}
    Temporal point processes offer a powerful framework for sampling from discrete distributions, yet they remain underutilized in existing literature. We show how to construct, for any target multivariate count distribution with downward-closed support, a multivariate temporal point process whose event-count vector in a fixed-length sliding window converges in distribution to the target as time tends to infinity. Structured as a system of potentially coupled infinite-server queues with deterministic service times, the sampler exhibits a discrete form of momentum that suppresses random-walk behaviour. The admissible families of processes permit both reversible and non-reversible dynamics. As an application, we derive a recurrent stochastic neural network whose dynamics implement sampling-based computation and exhibit some biologically plausible features, including relative refractory periods and oscillations. The introduction of auxiliary randomness reduces the sampler to a birth-death process, establishing the latter as a degenerate case with the same limiting distribution. In simulations on 63 target distributions, our sampler always outperforms these birth-death processes and frequently outperforms Zanella processes in multivariate effective sample size, with further gains when normalized by CPU time.
\end{abstract}

\maketitle

\section{Introduction}

The need to efficiently sample from complicated discrete distributions naturally arises in many disciplines, from Bayesian statistics to the computational sciences. Monte Carlo methods based on stochastic processes, particularly Markov chain Monte Carlo, are prominent in the literature, with Gibbs sampling and the Metropolis-Hastings algorithm being two of the best-known examples \citep{geman1984,metropolis1953,hastings1970}. Whilst sampling on continuous spaces has incorporated the use of continuous-time processes, with piecewise-deterministic Markov processes at the forefront \citep{bouchardcote2018,bierkens2019}, continuous-time samplers on discrete spaces have lagged behind significantly. A notable exception is the work of \cite{power2019}, in which both reversible and skew-reversible continuous-time Markov chains (CTMCs) were constructed to sample from multivariate discrete distributions.

Here, we present a method for sampling from a broad family of discrete distributions by simulation of a multivariate temporal point process. This yields a novel and efficient sampling algorithm with desirable properties, such as a discrete form of momentum and a choice between reversibility and non-reversibility. Specifically, our main theorem introduces a sampler for any multivariate count distribution with support \(R_Y \subseteq \mathbb{Z}_{\geq0}^d\) satisfying the condition
\begin{equation*}
    \{0,\dots,y_1\}\times\dots\times\{0,\dots,y_d\} \subseteq R_Y\quad\forall y \in R_Y\,.
\end{equation*}
That is, \(R_Y\) must be downward closed. We specify this target distribution by an unnormalized function \(f: \mathbb{Z}^d_{\geq 0} \rightarrow \mathbb{R}_{\geq0}\) and, without loss of generality, write the probability mass function (PMF) \(\pi: \mathbb{Z}^d_{\geq 0}\rightarrow [0, 1]\) as
\begin{equation}\label{eqn:mass}
    \pi(y) = \frac{f(y)}{Z}\prod_{i=1}^d\frac{1}{y_i!}\quad\forall y \in \mathbb{Z}^d_{\geq 0}\,,
\end{equation}
where \(Z\) is a normalizing constant.

Amongst stochastic processes, point processes have received relatively little attention as samplers. Beyond the elementary example of fixed-interval event counts from Poisson processes generating samples from Poisson distributions, we are aware only of work by \cite{buesing2011}. They developed multivariate temporal point processes to sample from distributions with support \(\{0,1\}^d\), with the aim of proposing a biologically plausible neural implementation of Boltzmann machine dynamics. Specifically, they constructed the sampler such that the limiting distribution of the event-count vector in a fixed-length sliding window is the target. Our work generalizes this pioneering approach significantly with respect both to the family of admissible target distributions and to the family of admissible point processes for each target.

In queueing-theoretic terms, our sampler applied to a \(d\)-dimensional target distribution can be described by \(d\) potentially coupled infinite-server queues with state- and time-dependent arrivals and deterministic service times. We express each queue with the notation \(\mathrm{M}_{Q, t}/\mathrm{D}/\infty\), where \(\mathrm{M}_{Q, t}\) denotes arrival rates as a function of time \(t\) and the state \(Q^-(t)\) just before the current time. The state contains the minimal amount of information needed to make the process Markovian, namely the length of each queue and the arrival times of the entities in each queue. For each target distribution, the sampler induces a family of admissible processes with different arrival rates. In equilibrium, the multivariate queue-length process is cyclostationary and first-order stationary; strict stationarity holds only in special cases. Whilst the queueing theory literature is extensive, to the best of our knowledge the limiting distributions of processes of this form have not been studied.

Closely related to queues are CTMCs, for which transition rates can be constructed from detailed balance or skew-detailed balance equations for a desired stationary distribution. Multivariate birth-death processes, a type of CTMC, are the processes most closely related to our work. In prior work it is often the case that multivariate birth-death processes are presented with the death rate on each component being proportional to that component's state \citep{barbour1988,eichelsbacher2008}. For example, it is known that a CTMC with birth rates \(\gamma(y\rightarrow y + e_i) = f(y + e_i)/f(y)\) and death rates \(\gamma(y\rightarrow y - e_i) = y_i\) on component \(i\), where \(e_i\) is a standard basis vector of \(\mathbb{R}^d\), will have the limiting distribution \(\pi\) in \cref{eqn:mass}. As we will demonstrate, CTMCs of this form can be relatively inefficient samplers compared to our method due to their tendency toward random-walk behaviour.

\section{Preliminaries}

Consider a multivariate point process \(X = (X_1,\dots,X_d)\), where \(X_1,\dots,X_d\) are \(d\) potentially dependent simple point processes on some Polish space \(\mathcal{X}\). More specifically, \(X_i = \{X_{i,1},X_{i,2},\dots\}\) is a random set of point locations in \(\mathcal{X}\) corresponding to the \(i\)th component of \(X\). Throughout this paper we will constrain ourselves to working with point processes that are locally finite: processes for which the number of points in any bounded subset of \(\mathcal{X}\) is almost surely finite. We can define the distribution of \(X\) on some bounded Borel set \(A\subseteq\mathcal{X}\) using local Janossy measures. This form of characterization of a point process is typical for both simple and marked point processes, but here we will provide a definition specific to the multivariate case.

Let \(K(y) = \{(i, k) : i = 1, \dots, d; k = 1, \dots y_i\}\) be the index set corresponding to a count vector \(y\in\mathbb{Z}_{\geq 0}^d\). For point counts \(Y = (\#(X_1\cap A),\dots,\#(X_d\cap A))\), the local Janossy measure of order \(y\) for \(X\) on \(A\) is defined as
\begin{equation*}
    J_{X,y}\Biggl(\prod_{(i,k)\in K(y)}B_{i,k} \Biggm\Vert A\Biggr) = \prod_{i=1}^d y_i!\pr\Biggl(Y=y, \bigcap_{i=1}^d\bigcap_{k=1}^{y_i}\{X_{i,k}\in B_{i,k}\}\Biggr)\,,
\end{equation*}
where \(X_{i,k}\in A\) and \(B_{i,k}\subseteq A\) for all \((i,k)\in K(y)\). Given that points on each component should be treated as indistinguishable from each other, \(J_{X,y}(\,\cdot\,\;\Vert\;A)\) assigns equal probabilities to all permutations of points on each component. Therefore, \(J_{X,y}(\,\cdot\,\;\Vert\;A)\) is symmetric in \(B_{1,1},\dots,B_{1,y_1}\), symmetric in \(B_{2,1},\dots,B_{2,y_2}\), and so on.

The local Janossy densities \(j_{X,y}(\,\cdot\,\;\Vert\;A)\), when they exist, are the densities of \(J_{X,y}(\,\cdot\,\;\Vert\;A)\) with respect to the Lebesgue measure \(\mu\). From here onwards, we will drop the \(y\) subscript as this will, at times, simplify notation greatly. Letting \(\mathcal{P}_{\mathrm{fin}}(A) = \{C\subseteq A : \#C<\infty\}\) be the set of all finite subsets of \(A\), the local Janossy densities of the point process \(X\) on \(A\) will be given by the function \(j_X(\,\cdot\,\;\Vert\;A):\mathcal{P}_{\mathrm{fin}}(A)^d\rightarrow\mathbb{R}_{\geq 0}\). To intuitively interpret these densities, let \(\mathcal{X} = \mathbb{R}\) and consider that \(j_X(\{x_{1,1},\dots,x_{1,y_1}\},\dots,\{x_{d,1},\dots,x_{d,y_d}\} \;\Vert\; A)\mu(\prod_{(i,k)\in K(y)} \mathrm{d}x_{i,k})\) gives the probability of observing exactly \(y_1\) points on the first component in \(A\) located in infinitesimal intervals around \(x_{1,1},\dots,x_{1,y_1}\in A\), exactly \(y_2\) points on the second component in \(A\) located in infinitesimal intervals around \(x_{2,1},\dots,x_{2,y_2}\in A\), and so on. Furthermore, given some Borel set \(B\subseteq A\), its complement \(B^c = A\setminus B\), and locations \(\hat{x}_1,\dots,\hat{x}_d \in \mathcal{P}_{\mathrm{fin}}(B^c)\), the local Janossy densities on \(B^c\) are provided by the marginalization
\begin{multline*}
    j_X(\hat{x}_1,\dots,\hat{x}_d\;\Vert\;B^c)=\\
    \sum_{y\in \mathbb{Z}_{\geq 0}^d}\int_{B^{\sum_{i=1}^d y_i}} \frac{j_X(\{x_{1,1},\dots,x_{1,y_1}\}\cup \hat{x}_1,\dots,\{x_{d,1},\dots,x_{d,y_d}\}\cup \hat{x}_d \;\Vert\; A)}{y_1!\cdots y_d!}\mu\Biggl(\prod_{(i,k)\in K(y)} \mathrm{d}x_{i,k}\Biggr)\,.
\end{multline*}

Finally, we also define conditional Janossy densities on \(B\) given a realization of the point process on \(B^c\), which are expressed as
\begin{equation*}
    j_X(x_1,\dots,x_d \mid \hat{x}_1,\dots,\hat{x}_d\;\Vert\;B \mid B^c) = \frac{j_X(x_1\cup \hat{x}_1,\dots,x_d\cup \hat{x}_d\;\Vert\;A)}{j_X(\hat{x}_1,\dots,\hat{x}_d\;\Vert\;B^c)}
\end{equation*}
for \(x_1,\dots,x_d \in \mathcal{P}_{\mathrm{fin}}(B)\) and \(\hat{x}_1,\dots,\hat{x}_d \in \mathcal{P}_{\mathrm{fin}}(B^c)\). The interpretation of these is the same as for the local Janossy densities, except we are conditioning on a configuration of points in a disjoint Borel set, assuming such a configuration is possible. Whilst the concept of a conditional Janossy density is not prevalent in the point process literature, it is one that will be of great assistance in the proof of our main theorem.

\section{Point processes for sampling}

Consider a multivariate temporal point process \(T = (T_1,\dots,T_d)\) on \(\mathbb{R}_{\geq 0}\), where \(T_i = \{T_{i,1},T_{i,2},\dots\}\) is a random set of points corresponding to the \(i\)th component of the process. Let \(N_i(t) = \#(T_i\cap(0,t])\) be the number of points that have occurred on component \(i\) up to, and including, time \(t \in \mathbb{R}_{\geq 0}\). We define the multivariate counting process \(N = (N(t))_{t \in \mathbb{R}_{\geq0}}\) associated with \(T\) by \(N(t) = (N_1(t),\dots,N_d(t))\). The multivariate conditional intensity function \(\lambda^* = (\lambda^*_1,\dots,\lambda^*_d)\) associated with this process is defined as
\begin{equation*}
    \lambda^*(t) = \lim_{\delta\downarrow0}\frac{\E(N[t + \delta] - N[t] \mid \mathcal{H}_{t-})}{\delta}\,,
\end{equation*}
where \(\mathcal{H}_{t-}\) is the process history up to, but not including, time \(t \in \mathbb{R}_{\geq 0}\). We will only ever consider temporal point processes that admit a conditional intensity process. Following convention, an asterisk superscript should be taken as shorthand, used to imply dependence on the process history. For example, we will frequently write \(\lambda^*\) instead of \(\lambda(\,\cdot\, \mid \mathcal{H}_{\,\cdot\,-})\). Conditional intensities will always be taken to be left-continuous where possible, ensuring uniqueness of the function for a given point process. See \cite{daley2003} for further discussion.

Choose \(m \in \mathbb{R}_{>0}\) to be a time scale for the point process. The precise role of \(m\) will be apparent shortly. Take \(M(t) = (t-m,t]\) and define the state of the sampler at time \(t \in \mathbb{R}_{\geq m}\) to be \(Q(t) = (Q_1(t),\dots,Q_d(t))\), where \(Q_i(t) = T_i\cap M(t)\) is the random set of point locations on component \(i\) in the preceding interval of length \(m\). Denote \(S_i(t) = \#Q_i(t)\ (i = 1, \dots, d)\) as the counts of these points. Define the difference process \(S = (S(t))_{t \in \mathbb{R}_{\geq m}}\) by \(S(t) = (S_1(t),\dots,S_d(t))\), such that it contains the counts of points in all intervals of length \(m\). The name of this process refers to the fact that \(S(t) = N(t) - N(t - m)\). Right-open and left-continuous versions of these objects are given by replacing \(M(t)\) with \(M^-(t) = [t-m,t)\). These versions will always be denoted by a negation superscript, such as \(Q^-_i(t) = T_i\cap M^-(t)\) and \(S^-_i(t) = \#Q^-_i(t)\). With respect to functions in general, \(-\) and \(+\) superscripts will denote left- and right-continuous variants respectively.

Take \(g:\mathcal{P}_{\mathrm{fin}}(\mathbb{R})^d\rightarrow\mathbb{R}_{\geq 0}\) to be a left-continuous, \(m\)-periodic function with the constraint
\begin{equation*}
    \int_{M(x)^{\sum_{i=1}^d y_i}}g(\{x_{1,1},\dots,x_{1,y_1}\},\dots,\{x_{d,1},\dots,x_{d,y_d}\})\mu\Biggl(\prod_{(i,k)\in K(y)} \mathrm{d}x_{i,k}\Biggr)=1\quad\forall y\in\mathbb{Z}_{\geq0}^d
\end{equation*}
for any fixed \(x\in\mathbb{R}\), and downward-closed support such that if \(g(x_1,\dots,x_d) > 0\) for some \(x_1,\dots,x_d\in\mathcal{P}_{\mathrm{fin}}(\mathbb{R})\), then \(g(\hat{x}_1,\dots,\hat{x}_d) > 0\) for all \(\hat{x}_1\subseteq x_1,\dots,\hat{x}_d\subseteq x_d\). These constraints imply that \(g(\varnothing,\dots,\varnothing) = 1\). With respect to Janossy densities and \(g\), we will occasionally use the terms ``right/left-continuous'', ``c\`{a}dl\`{a}g/c\`{a}gl\`{a}d'', and ``periodic''. These terms will signify that the relevant property applies to each element of every finite-set argument of the function. Here, this means \(g\) is left-continuous in \(x_{i,k}\) and returns the same value when \(m\) is added to \(x_{i,k}\) for any \((i,k)\in K(y)\) and any \(y\in\mathbb{Z}_{\geq0}^d\).

Assume there is some initial probability distribution for \(Q^-(m)\) such that \(\lambda^*(m)\) is defined. Consider a base intensity \(\eta^* = (\eta^*_1,\dots,\eta^*_d)\), defined by
\begin{equation*}
    \eta^*_i(t) = \frac{g(Q^-_1[t],\dots,Q^-_i[t]\cup\{t\},\dots,Q^-_d[t])}{g(Q^-_1[t],\dots,Q^-_d[t])}\quad(i = 1, \dots, d)\,.
\end{equation*}
For a process in equilibrium where \(\lambda^*(t) = \eta^*(t)\) for all \(t \in \mathbb{R}_{\geq m}\), the number of points on each component in any window of length \(m\) will be independent and Poisson-distributed with rate parameter \(1\). The distribution of point locations within that window is determined by \(g\). When \(\eta^*\) is independent of the process history, the process is simply composed of \(d\) independent Poisson processes with \(m\)-periodic intensity functions and \(\int_{M(t)}\lambda_i\mathrm{d}\mu = 1\ (i = 1, \dots, d)\). Sampling from \(\pi\) requires a simple modification to this base process.

\begin{theorem}\label{thm:sampler}
    Let \(f\) define \(\pi\) according to \cref{eqn:mass}. If
    \begin{equation*}
        \lambda^*_i(t) = \frac{f(S^-[t] + e_i)}{f(S^-[t])}\eta_i^*(t)\quad(i = 1, \dots, d)
    \end{equation*}
    for all \(t \in \mathbb{R}_{\geq m}\), then \(\pi\) is the limiting distribution of \(S(t)\) as \(t\uparrow\infty\).
\end{theorem}
\begin{proof}
    Consider an \(m\)-periodic multivariate point process \(X = (X_1,\dots,X_d)\) on \(\mathbb{R}\). Being \(m\)-periodic, a point exists at \(x\in\mathbb{R}\) if and only if a point also exists at \(x+m\) on the same component. Our proof strategy is to construct a sequence of Gibbs samplers with the distribution of \(X\) as the target, specified in terms of conditional Janossy densities, where the limit of the sequence gives the multivariate temporal point process of interest and its associated properties. We define the distribution of \(X\) using local Janossy densities on \(M(x)\) for some \(x\in\mathbb{R}\), such that the vector-valued number of points in this interval is \(\pi\)-distributed. This is achieved by letting these densities take the form
    \begin{multline*}
        j_X(\{x_{1,1},\dots,x_{1,y_1}\},\dots,\{x_{d,1},\dots,x_{d,y_d}\} \;\Vert\; M[x])\\
        \begin{aligned}
             & = \prod_{i=1}^d y_i!\pi(y)g(\{x_{1,1},\dots,x_{1,y_1}\},\dots,\{x_{d,1},\dots,x_{d,y_d}\}) \\
             & = \frac{f(y)}{Z}g(\{x_{1,1},\dots,x_{1,y_1}\},\dots,\{x_{d,1},\dots,x_{d,y_d}\})\,
        \end{aligned}
    \end{multline*}
    for any value of \(x\). Given the periodicity constraint, this local Janossy characterization completely defines the distribution of \(X\).

    Now consider a sequence of multivariate temporal point processes \((T^{(n)})_{n=2}^\infty\) on \(\mathbb{R}_{\geq 0}\), each with an associated multivariate counting process \(N^{(n)} = (N^{(n)}(t))_{t \in \mathbb{R}_{\geq0}}\) and multivariate conditional intensity function \(\lambda^{(n)*}\). For some time \(t \in \mathbb{R}_{\geq m}\), we define the state \(Q^{(n)}(t)\) by \(Q^{(n)}_i(t) = T^{(n)}_i\cap M(t)\ (i = 1, \dots, d)\). Each of these point processes will implement a systematic-scan Gibbs sampler for the distribution of \(X\), with \(n\) updates per sweep. Let \(\delta^{(n)} = m/n\) be a length of time and define the intervals \(P^{(n)}(t) = [t - m + \delta^{(n)},t)\) and \(F^{(n)}(t) = [t,t + \delta^{(n)})\). Let \((\tau^{(n)}_r)_{r=0}^\infty\), where \(\tau^{(n)}_r = m + r\delta^{(n)}\), be equally spaced points in time. For each value of \(r\), \(F^{(n)}(\tau^{(n)}_r)\) will be the interval on which the next Gibbs update is performed and \(P^{(n)}(\tau^{(n)}_r)\) the interval on which the update is conditioned. Assuming an admissible common initial distribution on \(M^-(m)\) across \(T\) and \(T^{(n)}\ (n = 2, 3, \dots)\), a Gibbs sampler for the distribution of \(X\) can be constructed by choosing to define the distribution of \(T^{(n)}\) such that
    \begin{equation*}
        j_{T^{(n)}}(\,\cdot\,\;\Vert\;F^{(n)}[\tau^{(n)}_r] \mid P^{(n)}[\tau^{(n)}_r]) = j_X(\,\cdot\,\;\Vert\;F^{(n)}[\tau^{(n)}_r] \mid P^{(n)}[\tau^{(n)}_r])\quad(r = 0, 1, \dots)\,.
    \end{equation*}

    We are interested in the behaviour of these point processes as \(n\uparrow\infty\). To establish convergence in distribution to the process \(T\), characterized by its conditional intensity function \(\lambda^*\), we demonstrate that \(\lambda^{(n)}(\,\cdot\, \mid \mathcal{H}_{\,\cdot\,-}) \rightarrow \lambda(\,\cdot\, \mid \mathcal{H}_{\,\cdot\,-})\) pointwise and in \(L^1_{\textrm{loc}}(\mathbb{R}_{\geq 0}, \mu)\) almost surely. As the conditional intensities are identical on \(M^-(m)\) for all processes, convergence is trivially satisfied on this initial interval.

    For some fixed \(t \in \mathbb{R}_{\geq m}\), consider the sequence \((\varepsilon^{(n)}(t))_{n=2}^\infty\) defined by \(\varepsilon^{(n)}(t) = \tau^{(n)}_r - t\), where for each \(n\), \(r\) is the unique index such that \(\tau^{(n)}_r \in (t, t + \delta^{(n)}]\). That is, the length of time from \(t\) to the end of the current Gibbs update interval is \(\varepsilon^{(n)}(t) \in (0, \delta^{(n)}]\). Let \(G^{(n)}(t) = [t - m + \varepsilon^{(n)}(t),t)\) be the minimal interval of the past that determines the conditional intensity at time \(t\). The conditional intensity functions for process \(n\) for \(t \in \mathbb{R}_{\geq m}\), given the process history \(\mathcal{H}^{(n)}_{t-}\), are
    \begin{align*}
        \lambda^{(n)*}_i(t) & = \lim_{\delta\downarrow 0}\frac{\pr(N^{(n)}[t + \delta] - N^{(n)}[t] = e_i \mid \mathcal{H}^{(n)}_{t-})}{\delta}\quad(i = 1, \dots, d)                                                                      \\
                            & = j_X(\varnothing,\dots,\varnothing,\underbrace{\{t\}}_{i\textrm{th argument}},\varnothing,\dots,\varnothing \mid T^{(n)}_1\cap G^{(n)}[t],\dots,T^{(n)}_d\cap G^{(n)}[t] \;\Vert\;\{t\} \mid G^{(n)}[t])\,.
    \end{align*}
    The limiting form of these conditional intensities for \(t \in \mathbb{R}_{\geq m}\) is then
    \ifpdf
        \begin{multline*}
            \lim_{n\uparrow\infty}\lambda^{(n)}_i(t \mid \mathcal{H}_{t-})\\
            \begin{aligned}
                 & = j_X(\varnothing,\dots,\varnothing,\underbrace{\{t\}}_{i\textrm{th argument}},\varnothing,\dots,\varnothing \mid T_1\cap (t-m, t),\dots,T_d\cap (t-m, t) \;\Vert\;\{t\} \mid (t-m, t)) \\
                 & = \frac{j_X(T_1\cap (t-m, t),\dots,T_i\cap (t-m, t)\cup\{t\},\dots,T_d\cap (t-m, t)\;\Vert\;M[t])}{j_X(T_1\cap (t-m, t),\dots,T_d\cap (t-m, t)\;\Vert\;M[t])}
            \end{aligned}
        \end{multline*}
        \begin{flalign*}
            \hspace{1.1925cm} & \overset{\textrm{a.s.}}{=} \frac{f(S^-[t] + e_i)}{f(S^-[t])}\frac{g(Q^-_1[t],\dots,Q^-_i[t]\cup\{t\},\dots,Q^-_d[t])}{g(Q^-_1[t],\dots,Q^-_d[t])} &  & \\
                              & = \lambda^*_i(t)\quad(i = 1, \dots, d)\,.                                                                                                         &  &
        \end{flalign*}
    \else
        \begin{multline*}
            \lim_{n\uparrow\infty}\lambda^{(n)}_i(t \mid \mathcal{H}_{t-})\\
            \begin{aligned}
                 & = j_X(\varnothing,\dots,\varnothing,\underbrace{\{t\}}_{i\textrm{th argument}},\varnothing,\dots,\varnothing \mid T_1\cap (t-m, t),\dots,T_d\cap (t-m, t) \;\Vert\;\{t\} \mid (t-m, t)) \\
                 & = \frac{j_X(T_1\cap (t-m, t),\dots,T_i\cap (t-m, t)\cup\{t\},\dots,T_d\cap (t-m, t)\;\Vert\;M[t])}{j_X(T_1\cap (t-m, t),\dots,T_d\cap (t-m, t)\;\Vert\;M[t])}                           \\
                 & \overset{\textrm{a.s.}}{=} \frac{f(S^-[t] + e_i)}{f(S^-[t])}\frac{g(Q^-_1[t],\dots,Q^-_i[t]\cup\{t\},\dots,Q^-_d[t])}{g(Q^-_1[t],\dots,Q^-_d[t])}                                       \\
                 & = \lambda^*_i(t)\quad(i = 1, \dots, d)\,.
            \end{aligned}
        \end{multline*}
    \fi
    These are precisely the Papangelou intensities for the distribution of \(X\), up to periodicity constraints. The restrictions on the support of \(\pi\) and \(g\) ensure that the denominators of the conditional intensities always remain greater than zero, even as points exit the sliding window at \(t-m\).

    Pointwise convergence of the conditional intensity functions alone is not sufficient to guarantee convergence in distribution of the corresponding point processes. However, by Corollary 4.46 of \cite{jacod2003}, their convergence in \(L^1_{\textrm{loc}}(\mathbb{R}_{\geq 0}, \mu)\) almost surely \textit{does} guarantee \(T^{(n)} \rightarrow T\) in distribution.  The following lemma ensures this, the proof of which we defer to \Cref{app:l1Proof}.
    \begin{lemma}\label{lem:l1}
        Almost surely, \(\lambda^{(n)}(\,\cdot\, \mid \mathcal{H}_{\,\cdot\,-}) \rightarrow \lambda(\,\cdot\, \mid \mathcal{H}_{\,\cdot\,-})\) in \(L^1_{\textrm{loc}}(\mathbb{R}_{\geq m}, \mu)\).
    \end{lemma}

    Having established convergence of the point processes, we can now pass properties of the processes in the sequence to the limiting process. As \((T^{(n)}, \varepsilon^{(n)}) \rightarrow (T, 0)\) in distribution and there are almost surely no points in \(T\) at any fixed time, the continuous mapping theorem with mappings \((u, v) \mapsto u_i\cap M(t + v[t])\ (i = 1, \dots, d)\) allows us to conclude that \(Q^{(n)}(t + \varepsilon^{(n)}[t]) \rightarrow Q(t)\) in distribution. It follows that, in equilibrium, if \(Q^{(n)}(t + \varepsilon^{(n)}[t])\sim j_X(\,\cdot\,\;\Vert\;M[t + \varepsilon^{(n)}\{t\}])\) for all \(t\), then \(Q(t) \sim j_X(\,\cdot\,\;\Vert\;M[t])\) for all \(t\). Thus, \(\pi\) is a stationary distribution of \(S\).

    Finally, consider a shifted skeleton chain \(\hat{Q} = (\hat{Q}(p))_{p=1}^\infty\), where \(\hat{Q}(p)\) is the state \(Q(pm)\) with all point locations shifted to \(M(m)\) by subtracting \((p-1)m\) from each. As the void probability \(\pr(S[t + m] = 0_d \mid \mathcal{H}_{t-}) = \exp(-\int_{M(t + m)}\sum_{i=1}^d \lambda^*_i\mathrm{d}\mu) > 0\) for all \(t\in\mathbb{R}_{\geq m}\), there is always positive probability that no new points will arrive in a time interval of length \(m\). Hence, with positive probability, we can move from any possible state \(\hat{Q}(p)\) to a state \(\hat{Q}(p + 1)\) with no points. As \(j_X(\,\cdot\,\;\Vert\; M[m])\) is a stationary distribution of \(\hat{Q}\) and the state with no points is accessible from any other state, making the singleton containing it an accessible atom, the chain is positive Harris recurrent. As we can remain in the state with no points over a single step, it is also aperiodic and therefore ergodic. It follows that \(\pi\) is the limiting distribution of \(S(t)\).
\end{proof}

The point process of \Cref{thm:sampler} yields familiar results in two simple cases, the first of which illustrates an application to a distribution with unbounded support.

\begin{example}[The Poisson Distribution and Poisson Processes]
    Consider a univariate PMF target given by the Poisson distribution with rate parameter \(\Lambda\),
    \begin{equation}\label{eqn:poissonTarget}
        \pi(y) = \frac{\Lambda^y \exp(-\Lambda)}{y!}\,,
    \end{equation}
    implying \(f(y) = \Lambda^y\) and \(Z = \exp(\Lambda)\). The point process sampler of \Cref{thm:sampler} for this target has conditional intensity function \(\lambda^*(t) = \eta^*(t)\Lambda^{S^-(t) + 1} / \Lambda^{S^-(t)} = \Lambda\eta^*(t)\). If \(\eta^*\) is independent of the process history, the result is a Poisson process with an \(m\)-periodic intensity function.
\end{example}

\begin{example}[Multivariate Binary Support and Boltzmann Machines]
    It is straightforward to verify that applying \Cref{thm:sampler} with constant \(\eta^*\) to distributions with support \(\{0,1\}^d\) yields the sampling method of \cite{buesing2011}.
\end{example}

Implementation of the point process sampler is relatively straightforward for constant \(\eta^*\), as shown in \Cref{alg:point} below. Literature on simulating inhomogeneous Poisson processes should be consulted for other choices of \(\eta^*\). For example, see \cite{cinlar2013} for the transformation method and \cite{lewis1979} for the thinning method. In the following, assume \(R_Y\) is non-singleton. Consider queues \(\mathcal{Q}_{\mathcal{T}}\) and \(\mathcal{Q}_{\mathcal{C}}\) to initially contain, enqueued in chronological order, the point times and components in \((0,m)\) respectively. The initial value for \(s\in R_Y\) should be the counts of these points. The notation \(\mathrm{Exp}(\xi)\) denotes an exponential distribution with mean \(\xi^{-1}\).
\begin{algo}[Point Process Sampler Simulation for Constant \(\eta^*\)]\label{alg:point}
    \hfill
    \begin{tabbing}
        \qquad Input point counts \(s\), point time queue \(\mathcal{Q}_{\mathcal{T}}\), point component queue \(\mathcal{Q}_{\mathcal{C}}\)\\
        \qquad Set \(t \leftarrow m\)\\
        \qquad Repeat \\
        \qquad \qquad Set \(\lambda^*_i \leftarrow m^{-1}f(s + e_i) / f(s)\ (i = 1, \dots, d)\)\\
        \qquad \qquad Sample a proposed time \(p\) to wait for a point, from \(\mathrm{Exp}(\sum_{i=1}^d \lambda^*_i)\)\\
        \qquad \qquad If \(\mathcal{Q}_{\mathcal{T}}\) is empty or \(t + p < \textsc{peek}(\mathcal{Q}_{\mathcal{T}}) + m\) then\\
        \qquad \qquad \qquad Sample a component index \(c\) with probability proportional to \(\lambda^*_c\)\\
        \qquad \qquad \qquad Set \(t\leftarrow t + p\)\\
        \qquad \qquad \qquad Set \(s\leftarrow s + e_c\)\\
        \qquad \qquad \qquad \(\textsc{enqueue}(\mathcal{Q}_{\mathcal{T}}, t)\)\\
        \qquad \qquad \qquad \(\textsc{enqueue}(\mathcal{Q}_{\mathcal{C}}, c)\)\\
        \qquad \qquad Else\\
        \qquad \qquad \qquad Set \(c\leftarrow\textsc{dequeue}(\mathcal{Q}_{\mathcal{C}})\)\\
        \qquad \qquad \qquad Set \(t\leftarrow\textsc{dequeue}(\mathcal{Q}_{\mathcal{T}}) + m\)\\
        \qquad \qquad \qquad Set \(s\leftarrow s - e_c\)
    \end{tabbing}
\end{algo}
The value of \(m\) has no bearing on the computational efficiency of the sampler, as it merely scales time, so it may simply be set to \(1\).

\section{Properties of the point process sampler}

\subsection{Time reversal}

In this subsection we will work with two-sided versions of the stochastic processes previously introduced. For simplicity, we will continue to use the same letters to denote these objects.

Given a two-sided extension of the difference process of \Cref{thm:sampler} in equilibrium, \(S = (S(t))_{t \in \mathbb{R}}\), we define the time-reversed process \(\bar{S} = (\bar{S}(t))_{t \in \mathbb{R}}\) by \(\bar{S}(t) = S^-(\tau - t)\) for some origin time \(\tau\in\mathbb{R}\). These processes are considered time reversible if and only if
\begin{equation}\label{eqn:reversal}
    (\bar{S}(t_1), \dots, \bar{S}(t_q)) \overset{\textrm{d}}{=} (S(t_1), \dots, S(t_q))\quad(q = 1, 2, \dots)
\end{equation}
for all \(\tau\) and all times \(t_1, \dots, t_q \in \mathbb{R}\). Time-reversibility also implies the condition of strict stationarity, rather than just first-order stationarity.

Denote all objects associated with this time-reversed difference process with the bar notation. \(\bar{T}\) is the resulting multivariate point process, with counting process \(\bar{N} = (\bar{N}(t))_{t \in \mathbb{R}}\), conditional intensity function \(\bar{\lambda}^*\), and sampler state \(\bar{Q}(t)\) at time \(t\). These objects have identical interpretations to their equivalents in the forward-time case. Given the time-reversed difference process, the corresponding point process can be expressed as
\begin{equation*}
    \bar{T} \overset{\textrm{a.s.}}{=} (\{t\in\mathbb{R} : \bar{S}(t) - \bar{S}^-(t) = e_1\}, \dots, \{t\in\mathbb{R} : \bar{S}(t) - \bar{S}^-(t) = e_d\})\,.
\end{equation*}
A two-sided count of points on component \(i\) at time \(t\), with respect to the convention \(\bar{N}_i(0) = 0\), is \(\bar{N}_i(t) = \#(\bar{T}_i\cap(0,t])\) for \(t\in\mathbb{R}_{\geq 0}\), and \(\bar{N}_i(t) = -\#(\bar{T}_i\cap(t,0])\) for \(t\in\mathbb{R}_{< 0}\). The conditional intensity function at time \(t\), given the process history \(\bar{\mathcal{H}}_{t-}\), is \(\bar{\lambda}^*(t) = \lim_{\delta\downarrow0}\E(\bar{N}[t + \delta] - \bar{N}[t] \mid \bar{\mathcal{H}}_{t-})/\delta\).
The state on component \(i\) at time \(t\) is \(\bar{Q}_i(t) = \bar{T}_i\cap M(t)\). The property \(\bar{S}_i(t) = \#\bar{Q}_i(t)\) then follows, implying the mapping between \(\bar{S}\) and \(\bar{T}\) is bijective up to null sets. A similar bijection exists between \(S\) and \(T\).

Define \(\mathcal{L}_a\) as an operator mapping between subsets of \(\mathbb{R}\), such that \(\mathcal{L}_aA = \{a - x : x \in A\}\) for some \(a\in\mathbb{R}\) and \(A\subseteq\mathbb{R}\). It should be emphasized that the reverse-time difference process is \textit{not} the process that occurs under the mappings \(\bar{T}_i = \mathcal{L}_\tau T_i\ (i = 1, \dots, d)\), as this would fail to account for the length of the sliding window. The following lemma addresses this potential source of confusion.
\begin{lemma}\label{lem:reversal}
    If \(\bar{S}(t) = S^-(\tau - t)\), then \(\bar{T}_i = \mathcal{L}_{\tau - m} T_i\ (i = 1, \dots, d)\).
\end{lemma}
\begin{proof}
    For the \(i\)th component, we have
    \begin{align*}
        \bar{S}_i(t)           & = S_i^-(\tau - t)                                     \\
        \#(\bar{T}_i\cap M[t]) & = \#(T_i\cap M^-[\tau - t])                           \\
                               & = \#(\mathcal{L}_{\tau - m}[T_i\cap M^-\{\tau - t\}]) \\
                               & = \#(\mathcal{L}_{\tau - m}T_i\cap M[t])\,.
    \end{align*}
    We know there is a bijection between \(T\) and \(S\), between \(S\) and \(\bar{S}\), and between \(\bar{S}\) and \(\bar{T}\). It follows that there is a bijection between \(T\) and \(\bar{T}\), and to satisfy the final line for all \(t\in\mathbb{R}\) the only possibility is that \(\bar{T}_i = \mathcal{L}_{\tau - m} T_i\).
\end{proof}

Equipped with these definitions and \Cref{lem:reversal}, we can introduce the following proposition, the proof of which is deferred to \Cref{app:reversalProof}. With \(\mathcal{P}_{R_Y}(\mathcal{X}) = \{x\in\mathcal{P}_{\mathrm{fin}}(\mathcal{X})^d : (\#x_1,\dots,\#x_d)\in R_Y\}\), we will write \(g\restriction_{\mathcal{P}_{R_Y}(\mathbb{R})}\) to denote the function \(g\) restricted to the domain \(\mathcal{P}_{R_Y}(\mathbb{R})\).
\begin{proposition}\label{pro:reversal}
    Suppose \(S\) has associated conditional intensity function \(\lambda^*\) for c\`{a}gl\`{a}d \(g\). Then \(\bar{S}\) has associated base intensities
    \begin{equation*}
        \bar{\eta}^*_i(t) = \frac{\bar{g}(\bar{Q}^-_1[t],\dots,\bar{Q}^-_i[t]\cup\{t\},\dots,\bar{Q}^-_d[t])}{\bar{g}(\bar{Q}^-_1[t],\dots,\bar{Q}^-_d[t])}\quad(i = 1, \dots, d)
    \end{equation*}
    and conditional intensities
    \begin{equation*}
        \bar{\lambda}^*_i(t) = \frac{f(\bar{S}^-[t] + e_i)}{f(\bar{S}^-[t])}\bar{\eta}^*_i(t)\quad(i = 1, \dots, d)\,,
    \end{equation*}
    where \(\bar{g}:\mathcal{P}_{\mathrm{fin}}(\mathbb{R})^d\rightarrow\mathbb{R}_{\geq 0}\) is defined by
    \begin{equation*}
        \bar{g}(x_1,\dots,x_d) = g^+(\mathcal{L}_\tau x_1,\dots,\mathcal{L}_\tau x_d)\,.
    \end{equation*}
    \(S\) and \(\bar{S}\) are time-reversible if and only if \(g\restriction_{\mathcal{P}_{R_Y}(\mathbb{R})} = \bar{g}\restriction_{\mathcal{P}_{R_Y}(\mathbb{R})}\) for all \(\tau \in \mathbb{R}\).
\end{proposition}
In addition to invariance under a global reversal of point locations, time-reversibility forces \(g\restriction_{\mathcal{P}_{R_Y}(\mathbb{R})}\) to depend only on the relative arrangement of points, not their absolute positions.

\subsection{Connection to sampling with birth-death processes}

It is well-established that a CTMC with birth rates \(\gamma(y\rightarrow y + e_i) = f(y + e_i)/f(y)\) and death rates \(\gamma(y\rightarrow y - e_i) = y_i\) on component \(i\) will have \(\pi\) as its limiting distribution. The birth transition rate resembles the conditional intensity function of the point process sampler. Here we show that such a birth-death sampler can be constructed by introducing an additional source of randomness to the point process sampler. This randomness effectively destroys point location information, resulting in death times that are independent of birth times.

Consider the sampler of \Cref{thm:sampler}, but from the perspective of \(d\) potentially coupled \(\mathrm{M}_{Q, t}/\mathrm{D}/\infty\) queues rather than a point process. From this viewpoint, \(S(t)\) gives the multivariate queue length, whilst the state \(Q(t)\) also contains the times of arrival to the queues. The process is Markovian with respect to this state. In this subsection we will consider the case where \(\eta^*\) is independent of the process history. We will make a single change to the process compared to that of \Cref{thm:sampler}. This change, which we refer to as a service-time redraw, transforms the process into one which is semi-Markov.

For each \(\varepsilon\in\mathbb{R}_{>0}\), consider a system of \(d\) potentially coupled queues. Let the state at time \(t \in \mathbb{R}_{\geq m}\), taking values in \(\mathcal{P}_{\mathrm{fin}}(M[t])^d\), be denoted by \(U^{(\varepsilon)}(t)\). The length of the \(i\)th queue is \(V_i^{(\varepsilon)}(t) = \#U_i^{(\varepsilon)}(t)\). Arrivals to the \(i\)th queue occur at the rate \(\eta_i(t)f(V_i^{(\varepsilon)-}[t] + e_i)/f(V_i^{(\varepsilon)-}[t])\). At times \(t\in\{m + r\varepsilon : r = 0, 1, \dots\}\) and for each queue \(i = 1, \dots, d\), the \(V^{(\varepsilon)-}_i(t)\) elements of \(U^{(\varepsilon)}_i(t)\) are redrawn to be independent of the process history and each other, distributed according to the probability density function \(\varphi_i(t,\,\cdot\,)\), where
\begin{equation*}
    \varphi_i(t, \hat{t}) = \begin{cases}
        \eta_i(\hat{t}) & \textrm{if } \hat{t}\in M^-(t) \\
        0               & \textrm{otherwise}
    \end{cases}\,.
\end{equation*}
As \(\varepsilon \downarrow 0\), the frequency of these service-time redraws increases arbitrarily, yielding the following result. The proof is deferred to \Cref{app:bridgingProof}.
\begin{proposition}\label{pro:bridging}
    As \(\varepsilon\downarrow 0\), the queue-length process \(V^{(\varepsilon)} = (V^{(\varepsilon)}(t))_{t\in\mathbb{R}_{\geq m}}\) converges in distribution to a time-inhomogeneous birth-death process \(V = (V(t))_{t\in\mathbb{R}_{\geq m}}\) with birth rates
    \begin{equation*}
        \gamma_t(y\rightarrow y + e_i) = \frac{f(y + e_i)}{f(y)}\eta_i(t)\quad(i = 1, \dots, d)
    \end{equation*}
    and death rates
    \begin{equation*}
        \gamma_t(y\rightarrow y - e_i) = y_i\eta_i(t)\quad(i = 1, \dots, d)\,.
    \end{equation*}
    Furthermore, \(V(t)\) retains the limiting distribution \(\pi\) as \(t\uparrow\infty\).
\end{proposition}

With the point process sampler taking the form of potentially coupled \(\mathrm{M}_{Q, t}/\mathrm{D}/\infty\) queues and the birth-death sampler taking the form of potentially coupled \(\mathrm{M}_{Q, t}/\mathrm{M}_t/\infty\) queues, we conjecture that the insensitivity property holds and more general potentially coupled \(\mathrm{M}_{Q, t}/\mathrm{G}_t/\infty\) queues exist which have \(\pi\) as their limiting distribution. See \cite{burman1981} and \cite{zachary2007} for further information on the insensitivity property in queues.

\section{Application: a stochastic neural network model}

It has been hypothesized that neural systems learn to represent probability distributions and use them to make probabilistic inferences about unknown quantities; however, few detailed candidate models for this process exist. One school of thought proposes that these computations may depend on a form of sampling \citep{hoyer2002,fiser2010,berkes2011,buesing2011,pecevski2011,habenschuss2013,hennequin2014,aitchison2016,haefner2016,orban2016,petrovici2016,echeveste2020,festa2021,korcsakgorzo2022,terada2024}. Although by no means a complete biological description, the point process sampler introduced here may offer one route to such a model.

Most neurons communicate by discrete ``firing'' events, the timings of which may be modelled as realizations of a point process. We consider the possibility that the numbers of events within a time window of length \(m\) across a set of \(d\) neurons correspond to a sample from the relevant distribution. Consider distributions on \(R_Y = \mathbb{Z}_{\geq0}^d\) with PMFs \(\pi\) defined according to \cref{eqn:mass} by the functions
\begin{equation}\label{eqn:neuralTarget}
    f(y) = \exp\Biggl(\frac{1}{2}y^TWy + \biggl[b - \frac{1}{2}\operatorname{diag}\{W\}\biggr]^Ty - \sum_{i = 1}^d \frac{\exp(a_1 y_i + a_0)}{\exp(a_1) - 1}\Biggr)
\end{equation}
for symmetric ``weight'' matrix \(W \in \mathbb{R}^{d\times d}\), ``bias'' vector \(b \in \mathbb{R}^d\), and ``refractory'' coefficients \(a_0\in\mathbb{R}\) and \(a_1 \in\mathbb{R}_{>0}\). This form is a generalization to non-negative integers of the Boltzmann machine, or Sherrington-Kirkpatrick model, discussed by \cite{buesing2011}. For fixed \(a_1\), it describes an exponential family of distributions.

By \Cref{thm:sampler}, the counts \(S(t)\) will represent samples from \(\pi\) as \(t\uparrow\infty\) provided they are based on point processes with conditional intensities
\begin{equation}\label{eqn:neuralIntensities}
    \lambda_i^*(t) = \exp([WS^-\{t\} + b]_i - \exp[a_1 S_i^-\{t\} + a_0])\eta^*_i(t)\quad(i = 1, \dots, d)\,.
\end{equation}
Interpreting these intensities as the firing rates of a population of \(d\) neurons yields a stochastic neural network where the matrix entries \(w_{i,j}\) correspond to the strength of connections, or synapses, between neurons. To implement these rates, each neuron must integrate its own firing events as well as incoming events from other neurons within a window of length \(m\). This integration window may correspond to the period of an oscillatory input through the function \(\eta^*\).

The double exponential term in the conditional intensities ensures that \(\pi\) is a valid PMF. In the neural context it enforces a relative refractory period, meaning that a neuron is less likely to fire if it has already done so recently. This overcomes a limitation of the continuous-time neural sampling method introduced by \cite{buesing2011}, where the limiting distribution of the process was no longer provably the target distribution once a biologically necessary relative refractory period was introduced.

As a model of neural inference and learning, we partition the neurons, and therefore the random vector \(Y\sim\pi\), into two groups: observed neurons \(O\), whose state is set by external input, and latent neurons \(L\). Inference involves computation of \(\pi_{L\mid O}\), which is achieved approximately by drawing events according to the intensities of \cref{eqn:neuralIntensities} in the latent neurons, conditioned on the state of the observed neurons. Learning involves updating the parameters \(\Theta = \{w_{i,j} : i = 1, \dots, d; j = i, \dots, d\} \cup \{b_i : i = 1, \dots, d\}\) so that the marginal distribution \(\pi_O\) of the observed neurons \(O\) matches an empirical data distribution, here represented by \(\psi\). The maximum likelihood parameters can be found by following the negated gradient of the KL-divergence \(D_{\mathrm{KL}}(\psi\;\Vert\;\pi_O)\) with respect to parameters. For each \(\theta \in \Theta\), the gradient component is
\begin{equation*}
    \diffp{}{\theta} D_{\mathrm{KL}}(\psi\;\Vert\;\pi_O) = \E_{Y\sim \pi}\biggl(\diffp{}{\theta} \log[f\{Y\}]\biggr) - \E_{\substack{O\sim \psi \\ (L \mid O = o) \sim \pi_{L\mid O}(\,\cdot\, \mid o)}}\biggl(\diffp{}{\theta} \log[f\{Y\}]\biggr)\,,
\end{equation*}
where
\begin{equation*}
    \diffp{}{\theta} \log(f[Y]) =
    \begin{cases}
        Y_iY_j           & \textrm{if } \theta = w_{i,j} \\
        Y_i(Y_i - 1) / 2 & \textrm{if } \theta = w_{i,i} \\
        Y_i              & \textrm{if } \theta = b_i
    \end{cases}\quad(i = 1, \dots, d; j = i + 1, \dots, d)\,.
\end{equation*}
Synaptic update rules of this form are referred to as contrastive Hebbian rules and have been widely studied in neuroscience.

The model is not entirely biologically plausible. In particular, the requirements that connection weights be symmetric and that events be integrated within a single equally weighted, rather than decaying, memory window are not obviously met by biological systems. Nonetheless, it may offer a promising starting point for further development.

\section{Sampling efficiency simulation study}

\subsection{Study design}

We analyse the efficiency of the proposed point process sampler relative to other jump process samplers. This comparison is made in terms of the multivariate effective sample size \textsc{ess} introduced by \cite{vats2019} and the multivariate effective sample size per CPU second \textsc{ess/second}. For the \(d\)-dimensional target distribution \(\pi\) with covariance matrix \(\Xi\) and sampler with asymptotic covariance matrix \(\Sigma\), the \textsc{ess} for \(k\) sequential samples is
\begin{equation*}
    \textsc{ess} = k\biggl(\frac{\det(\Xi)}{\det(\Sigma)}\biggr)^{1/d}\,.
\end{equation*}
As the samplers are operating in continuous time, \(k\) represents the number of sequential \textit{weighted} samples generated. More specifically, we generate \(k\) sequential samples from the embedded jump process, each weighted by their holding times.

We compare the point process sampler of \Cref{thm:sampler} to the birth-death sampler of \Cref{pro:bridging} and to Zanella process samplers \citep{power2019} with three different balancing functions. For the point process and birth-death samplers we let \(m = 1\), although this choice has no impact on the \textsc{ess}, and let \(\eta_i(t) = 1\ (i = 1, \dots, d)\). For the Zanella process samplers we restrict neighbouring states to those which differ by a standard basis vector, and we use the balancing functions \(z \mapsto z^{1/2}\), \(z \mapsto \min(1, z)\), and \(z \mapsto z/(1 + z)\).

For all samplers we run each chain for \(1\,000\,000\) steps of burn-in, followed by an additional \(k = 9\,000\,000\) steps from which we estimate \(\Xi\) and \(\Sigma\). The asymptotic covariance matrix \(\Sigma\) is estimated via batch means with a batch size of \(3\,000\). For estimating the \textsc{ess/second}, we record the CPU time required to simulate those \(9\,000\,000\) steps, ensuring we obtain as fair a comparison between the samplers as reasonably possible. For every combination of the 63 target distributions and 5 samplers, each simulation is run 10 times to obtain estimates with 95\% confidence intervals. The initial state/queue length of each chain is a zero vector. \Cref{fig:essPlots} presents the results. Values of \textsc{ess} are scaled by \(1/9\,000\) so that the plots present \textsc{ess} for \(1\,000\) sequential samples. See \Cref{app:experiments} for more details about the experiments.

\begin{figure}[b]
    \centering
    \includegraphics[width=\linewidth, alt={Plots showing the efficiency of all tested jump process samplers across all tested target distributions. Each distribution type (Poisson, Sherrington-Kirkpatrick, and stochastic neural network) is tested with 21 different scaling values (lambda, beta, and alpha respectively).}]{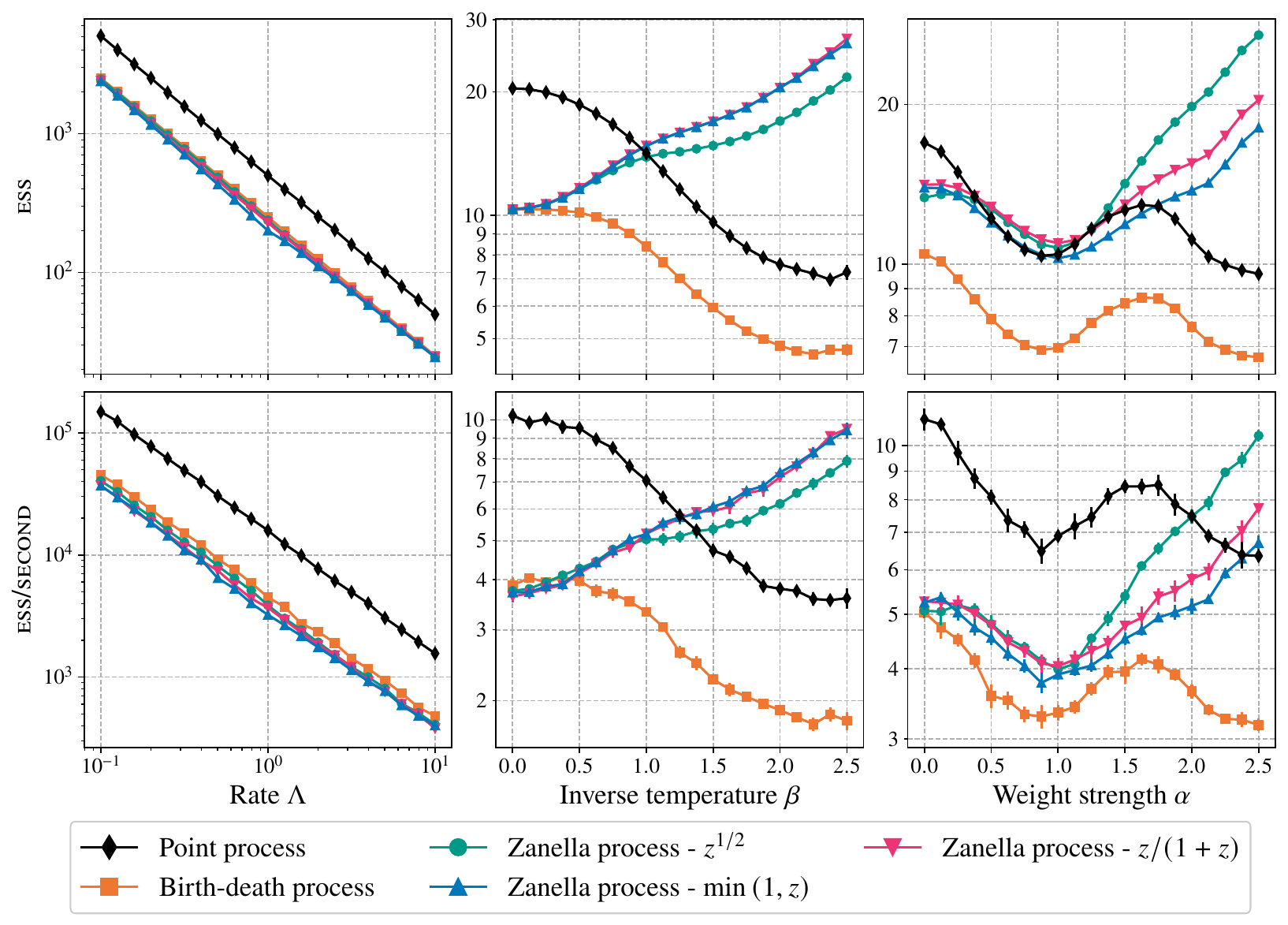}
    \caption{Experimentally derived efficiency estimates for five different jump process samplers applied to 63 different target distributions. Error bars show 95\% confidence intervals, but in most cases they are smaller than the markers and hence not visible. Left column: Poisson distribution (log-log plots). Middle column: Sherrington-Kirkpatrick model (log-linear plots). Right column: stochastic neural network model (log-linear plots). Top row: \textsc{ess} for \(1\,000\) sequential weighted samples. Bottom row: \textsc{ess} per CPU second.}
    \label{fig:essPlots}
\end{figure}

\cite{power2019} also introduce three skew-reversible CTMC samplers: the tabu sampler, the discrete coordinate sampler, and the discrete zig-zag process. These samplers attempt to exploit the algebraic structure available for a given state space, incorporating group actions and lifting to help minimize backtracking. The applicability and benefit of these samplers are highly dependent on the target distribution. We do not test against these for two reasons: (a) no single sampler is suitable for all tested target distributions, and (b) we employ a constant base intensity \(\eta^*\), representing the most basic implementation of our framework, making the likewise reversible Zanella processes the most appropriate benchmarks. We therefore treat Zanella processes as the baseline for experimental comparison, just as in the experimental work of \cite{power2019}.

\subsection{Poisson distribution}

The simplest distribution we evaluate the samplers on is the Poisson distribution of \cref{eqn:poissonTarget}. Whilst there is little interest in using a stochastic process to sample from such a distribution, its simplicity makes it a case where the point process sampler shows a clear advantage with easily interpretable empirical results. The simulations highlight two important points.

First, given that the point process sampler simply results in a Poisson process in this case, \(S(t)\) and \(S(t + m)\) will be independent for all \(t\in\mathbb{R}_{\geq m}\). This is unlike the other samplers, whose transition rates are state-dependent. This provides the point process sampler with a large advantage in \textsc{ess}. This suggests that the point process sampler should perform well on distributions where the resulting point timings are independent or weakly dependent.

Second, \textsc{ess} drops as \(\Lambda\) grows. The more points the point process sampler needs to generate in a window \(m\) time units long, the more inefficiently it will perform. Hence, we suggest that the sampler may perform well for distributions with most of the probability mass concentrated around small numerical values, but more inefficiently with it concentrated around large numerical values. Whilst the other samplers behave similarly in our simulations, we note that we have restricted their jumps to neighbouring states. Removing this restriction could lead to improved results.

\subsection{Sherrington-Kirkpatrick model}

The Sherrington-Kirkpatrick model, also referred to as a fully connected Ising model, is a widely studied model of spin interactions \citep{sherrington1975}. It consists of \(d\) atoms, whose individual spin can be in one of two states, with interactions between all pairs of atoms. An interaction between two atoms can be ferromagnetic, favouring them to take identical spins, or antiferromagnetic, favouring them to take opposite spins. In the absence of an external magnetic field, the spin configuration probabilities are proportional to \(\exp(-\beta H[\sigma])\) for Hamiltonian
\begin{equation*}
    H(\sigma) = -\sum_{i = 1}^d\sum_{j > i}^d J_{i,j}\sigma_i\sigma_j\,,
\end{equation*}
where \(\sigma \in \{-1, 1\}^d\) is the spin configuration, \(J_{i,j} \in \mathbb{R}\) is the interaction between atoms \(i\) and \(j\), and \(\beta\in\mathbb{R}_{\geq 0}\) is the inverse temperature. Often, the interactions are independently sampled from \(\mathcal{N}(0, 1/d)\).

Clearly, this distribution does not have the correct support for our sampler. To account for this, we apply the simple bijective mapping \(\sigma_i \mapsto (\sigma_i + 1) / 2\ (i = 1, \dots, d)\). This results in a PMF with support \(\{0, 1\}^d\) and \(f\) given by
\begin{equation*}
    f(y) = \begin{cases}
        \exp(\beta[y^TWy - b^Ty]) & \textrm{if } y\in \{0, 1\}^d \\
        0                         & \textrm{otherwise}
    \end{cases}\,,
\end{equation*}
where \(b_i = \sum_{j=1}^d W_{i,j}\) and \(W\) is a symmetric matrix with zeros on the main diagonal and entries above/below the main diagonal sampled independently from \(\mathcal{N}(0, 4/d)\). In our simulations we use dimensionality \(d = 100\) and a common \(W\) across all runs. The results indicate that, for this transformed distribution, the point process sampler outperforms existing jump process samplers in terms of \textsc{ess} up until around \(\beta = 1\), after which all three Zanella samplers are superior. This suggests that, at least for constant \(\eta^*\), the point process sampler performs best in the weak-coupling regime.

\subsection{Stochastic neural network model}

Here we consider a multivariate distribution with support \(\mathbb{Z}_{\geq0}^d\). Specifically, we consider the stochastic neural network model of \cref{eqn:neuralTarget} with dimensionality \(d = 100\), refractory coefficients \(a_0 = 0\) and \(a_1 = 1\), bias terms \(b_i = 5\ (i = 1, \dots, d)\), and symmetric weight matrix \(W = \alpha\hat{W}\) with the entries of \(\hat{W}\) sampled independently from \(\mathcal{N}(0, 1/d)\). The matrix \(\hat{W}\) is common to all runs. As with the Sherrington-Kirkpatrick model, we find there to be a region of small weight strength \(\alpha\in\mathbb{R}_{\geq 0}\) where the point process sampler outperforms the other jump process samplers in terms of \textsc{ess}, but this advantage is lost as \(\alpha\) increases. The major advantage of the point process sampler comes from its execution speed, with the \textsc{ess/second} being notably higher than the competing samplers for about 80\% of the range of weight strengths tested.

\subsection{Commonly observed behaviour}

Across all 63 distributions sampled from, we observe the \textsc{ess} of the point process sampler to be 1.4 to 2.0 times that of the birth-death sampler. Given that \Cref{pro:bridging} tells us the birth-death sampler can be constructed by introducing a large amount of uninformative randomness to the point process sampler, these results are not surprising.

It appears as though the point process sampler possesses some sense of ``momentum'', whilst the birth-death sampler does not. As points must remain in memory for exactly \(m\) units of time, the point process sampler cannot immediately backtrack in many circumstances. Take, as a simple example, a 2-dimensional target distribution with support \(\{0, 1\}^2\). If the sampler makes the jump in queue length \((1, 0)\rightarrow(1, 1)\), it is now forced to jump to \((0, 1)\), as points must drop out of memory in the same order they appear. If \((0, 1)\) is of relatively high probability compared to \((1, 1)\), there is a good chance the next jump will be to \((0, 0)\). Otherwise, the sampler will take the trajectory \((0, 1)\rightarrow(1, 1)\rightarrow(1, 0)\), effectively reversing the direction of the momentum. For these reasons, we conjecture that the point process sampler may always have a greater \textsc{ess} than the birth-death sampler, regardless of the target distribution.

When also considering computational cost, the \textsc{ess/second} of the point process sampler rises to 1.9 to 3.6 times that of the birth-death sampler. For any \(d\)-dimensional target distribution, the CTMC samplers all require the computation and use of \(2d\) transition rates, whereas the point process sampler only needs \(d\) conditional intensities. The ability to store point times in a queue and rapidly check to see if a point has fallen out of memory provides the point process sampler with a significant boost in efficiency.

\section{Discussion}

The introduction of the point process sampler brings about new open questions. Some of these questions relate to the application of the method, such as its relevance to modelling neural dynamics. Other questions pertain to the theoretical properties of the method.

One such question that we have yet to address is, ``What is the effect of imposing oscillatory dynamics through the choice of \(g\) on sampling efficiency?'' We conjecture that the choice of \(g\) may have an impact on metrics of efficiency such as mixing time and asymptotic variance. Take, for example, the case where \(\eta^*\) is independent of the process history. If \(\eta_1 = \dots = \eta_d\), we do not expect there to be any meaningful impact on the sampler, as this choice of \(\eta\) can simply be viewed as a nonlinear scaling of time compared to the non-oscillatory case, pushing points together where \(\eta\) is large in magnitude and spreading them out where \(\eta\) is small in magnitude. However, other choices for \(\eta\) cannot be viewed in the same time-scaling manner. Taking inspiration from \cite{he2016} who found that scan order in Gibbs sampling can have a significant impact on mixing times, we suspect that \(\eta^*\) may play a similar role by allowing conditional intensities for different components to vary in magnitude scaling relative to one another. We consider the design of an optimal \(\eta^*\) for a given target to be a mathematically challenging problem that lies beyond the scope of this work.

\section*{Declaration of the use of generative AI and AI-assisted technologies}

During the preparation of this work the authors used ChatGPT and Google Gemini for proofreading and occasional sentence rephrasing. After using these tools/services the authors reviewed and edited the content as necessary and take full responsibility for the content of the publication.

\newpage

\appendix

\crefalias{section}{crefAppendix}
\crefalias{subsection}{crefAppendix}
\crefalias{subsubsection}{crefAppendix}

{\fontsize{10}{12}\selectfont
    \section{Proofs}

    \subsection{Proof of \texorpdfstring{\Cref{lem:l1}}{Lemma \ref{lem:l1}}}\label{app:l1Proof}

    \setcounter{lemma}{0}
    \begin{lemma}[Restated]
        Almost surely, \(\lambda^{(n)}(\,\cdot\, \mid \mathcal{H}_{\,\cdot\,-}) \rightarrow \lambda(\,\cdot\, \mid \mathcal{H}_{\,\cdot\,-})\) in \(L^1_{\textrm{loc}}(\mathbb{R}_{\geq m}, \mu)\).
    \end{lemma}

    \begin{proof}
        Demonstrating convergence in \(L^1_{\textrm{loc}}(\mathbb{R}_{\geq m}, \mu)\) almost surely can be achieved with Lebesgue's dominated convergence theorem applied to all \(d\) components separately, which is the approach we take here. To aid readability in what follows, let \(\lvert y\rvert = \sum_{i=1}^d y_i\), \(G^{(n)c}(t) = (t-m,t)\setminus G^{(n)}(t)\), \(C^{(n)}_i(t) = T_i\cap G^{(n)}(t)\ (i = 1, \dots, d)\), and \(t_i = \{t_{i,1},\dots,t_{i,y_i}\}\ (i = 1, \dots, d)\). For each component \(i = 1, \dots, d\), we find the dominating function \(h_i(\,\cdot\, \mid \mathcal{H}_{\,\cdot\,-})\):
        \begin{multline*}
            \lambda^{(n)}_i(t \mid \mathcal{H}_{t-}) = j_X(\varnothing,\dots,\varnothing,\overbrace{\{t\}}^{i\textrm{th argument}},\varnothing,\dots,\varnothing \mid C^{(n)}_1[t],\dots,C^{(n)}_d[t] \;\Vert\;\{t\} \mid G^{(n)}[t])\\
            \begin{aligned}
                 & = \frac{\sum_{y\in \mathbb{Z}_{\geq 0}^d}\int_{G^{(n)c}(t)^{\lvert y\rvert}} \frac{j_X(t_1\cup C^{(n)}_1[t],\dots,t_i\cup C^{(n)}_i[t] \cup \{t\},\dots,t_d\cup C^{(n)}_d[t]\;\Vert\;M[t])}{y_1!\cdots y_d!}\mu(\prod_{(i,k)\in K(y)} \mathrm{d}t_{i,k})}{\sum_{y\in \mathbb{Z}_{\geq 0}^d}\int_{G^{(n)c}(t)^{\lvert y\rvert}} \frac{j_X(t_1\cup C^{(n)}_1[t],\dots,t_d\cup C^{(n)}_d[t]\;\Vert\;M[t])}{y_1!\cdots y_d!}\mu(\prod_{(i,k)\in K(y)} \mathrm{d}t_{i,k})} \\
                 & \leq \sum_{y\in \mathbb{Z}_{\geq 0}^d}\int_{M(t)^{\lvert y\rvert}}\frac{j_X(t_1\cup C^{(n)}_1[t],\dots,t_i\cup C^{(n)}_i[t] \cup \{t\},\dots,t_d\cup C^{(n)}_d[t]\;\Vert\;M[t])}{y_1!\cdots y_d!j_X(C^{(n)}_1[t],\dots,C^{(n)}_d[t] \;\Vert\; M[t])}\mu\Biggl(\prod_{(i,k)\in K(y)} \mathrm{d}t_{i,k}\Biggr)                                                                                                                                                          \\
                 & \leq \max_{\substack{{D_1\subseteq T_1\cap (t-m,t)}\\{\vphantom{\int\limits^x}\smash{\vdots}}\\{D_d\subseteq T_d\cap (t-m,t)}}}\sum_{y\in \mathbb{Z}_{\geq 0}^d}\int_{M(t)^{\lvert y\rvert}} \frac{j_X(t_1\cup D_1,\dots,t_i\cup D_i \cup \{t\},\dots,t_d\cup D_d\;\Vert\;M[t])}{y_1!\cdots y_d!j_X(D_1,\dots,D_d\;\Vert\;M[t])}\mu\Biggl(\prod_{(i,k)\in K(y)} \mathrm{d}t_{i,k}\Biggr)                                  \\
                 & = h_i(t \mid \mathcal{H}_{t-})\,.
            \end{aligned}
        \end{multline*}
        The \(j_X(D_1,\dots,D_d\;\Vert\;M[t])\) term in \(h_i\) is piecewise constant in \(t\) and non-zero almost surely. The \(j_X(t_1\cup D_1,\dots,t_i\cup D_i \cup \{t\},\dots,t_d\cup D_d\;\Vert\;M[t])\) term in \(h_i\) is locally integrable in \(t\) as a result of \(g\) being locally integrable. The maximum is taken over finitely many point configurations, the collection of which only changes with \(t\) when a point enters or leaves \((t-m,t)\), and the maximum of finitely many locally integrable functions is also locally integrable. Therefore, \(h_i(\,\cdot\, \mid \mathcal{H}_{\,\cdot\,-})\ (i = 1, \dots, d)\) are almost surely locally integrable and, by Lebesgue's dominated convergence theorem, \(\lambda^{(n)}(\,\cdot\, \mid \mathcal{H}_{\,\cdot\,-}) \rightarrow \lambda(\,\cdot\, \mid \mathcal{H}_{\,\cdot\,-})\) in \(L^1_{\textrm{loc}}(\mathbb{R}_{\geq m}, \mu)\) almost surely.
    \end{proof}

    \subsection{Proof of \texorpdfstring{\Cref{pro:reversal}}{Proposition \ref{pro:reversal}}}\label{app:reversalProof}

    Take all stochastic processes in this subsection to be two-sided and in equilibrium.

    \setcounter{proposition}{0}
    \begin{proposition}[Restated]
        Suppose \(S\) has associated conditional intensity function \(\lambda^*\) for c\`{a}gl\`{a}d \(g\). Then \(\bar{S}\) has associated base intensities
        \begin{equation*}
            \bar{\eta}^*_i(t) = \frac{\bar{g}(\bar{Q}^-_1[t],\dots,\bar{Q}^-_i[t]\cup\{t\},\dots,\bar{Q}^-_d[t])}{\bar{g}(\bar{Q}^-_1[t],\dots,\bar{Q}^-_d[t])}\quad(i = 1, \dots, d)
        \end{equation*}
        and conditional intensities
        \begin{equation*}
            \bar{\lambda}^*_i(t) = \frac{f(\bar{S}^-[t] + e_i)}{f(\bar{S}^-[t])}\bar{\eta}^*_i(t)\quad(i = 1, \dots, d)\,,
        \end{equation*}
        where \(\bar{g}:\mathcal{P}_{\mathrm{fin}}(\mathbb{R})^d\rightarrow\mathbb{R}_{\geq 0}\) is defined by
        \begin{equation*}
            \bar{g}(x_1,\dots,x_d) = g^+(\mathcal{L}_\tau x_1,\dots,\mathcal{L}_\tau x_d)\,.
        \end{equation*}
        \(S\) and \(\bar{S}\) are time-reversible if and only if \(g\restriction_{\mathcal{P}_{R_Y}(\mathbb{R})} = \bar{g}\restriction_{\mathcal{P}_{R_Y}(\mathbb{R})}\) for all \(\tau \in \mathbb{R}\).
    \end{proposition}

    \begin{proof}
        {Let \(\mathcal{F}_{t+}\) be the future of \(T\) from, but not including, time \(t\). By recognizing the \(m\)-periodic structure of the local Janossy densities of \(X\) from the proof of \Cref{thm:sampler}, it is possible to write an instantaneous detailed balance equation for the point process sampler in equilibrium at time \(t\),
        \begin{multline*}
            j_X(t_1 \cup \check{t}_1,\dots,t_d \cup \check{t}_d\;\Vert\;M^-[t])j_X(\hat{t}_1,\dots,\hat{t}_d \mid t_1,\dots,t_d\;\Vert\;\{t\} \mid (t-m,t))\\
            = j_X(t_1 \cup \hat{t}_1,\dots,t_d \cup \hat{t}_d\;\Vert\;M[t])j_X(\check{t}_1,\dots,\check{t}_d \mid t_1,\dots,t_d\;\Vert\;\{t-m\} \mid (t-m,t))\,,
        \end{multline*}
        where \(t_1,\dots,t_d \in \mathcal{P}_{\mathrm{fin}}((t-m,t))\), \(\hat{t}_1,\dots,\hat{t}_d \in \mathcal{P}_{\mathrm{fin}}(\{t\})\), and \(\check{t}_1,\dots,\check{t}_d \in \mathcal{P}_{\mathrm{fin}}(\{t-m\})\). In highlighting the fact that
        \begin{multline*}
            \lambda_i(t \mid \mathcal{H}_{t-})\\
            = j_X(\varnothing,\dots,\varnothing,\underbrace{\{t\}}_{i\textrm{th argument}},\varnothing,\dots,\varnothing \mid T_1\cap (t-m, t),\dots,T_d\cap (t-m, t) \;\Vert\;\{t\} \mid (t-m, t))\quad(i = 1, \dots, d)\,,
        \end{multline*}
        it immediately becomes clear from instantaneous detailed balance that
        \begin{multline*}
            \lambda_i(t-m \mid \mathcal{F}_{(t-m)+})\\
            = j_X(\varnothing,\dots,\varnothing,\underbrace{\{t-m\}}_{i\textrm{th argument}},\varnothing,\dots,\varnothing \mid T_1\cap (t-m, t),\dots,T_d\cap (t-m, t)\;\Vert\;\{t-m\} \mid (t-m, t))\quad(i = 1, \dots, d)
        \end{multline*}
        are conditional intensities for \(T\) that depend on the process future and tell us the rate at which points occur backwards in time. Under the time mapping \(t \mapsto \tau - m - t\) obtained by \Cref{lem:reversal}, we deduce that the conditional intensities for \(\bar{T}\) are
        \begin{multline*}
            \bar{\lambda}_i(t \mid \bar{\mathcal{H}}_{t-}) = \lambda^+_i(\tau-m-t \mid \mathcal{F}_{(\tau-m-t)+})\quad(i = 1, \dots, d)\\
            \begin{aligned}
                 & = \begin{multlined}
                         j_X^+(\varnothing,\dots,\varnothing,\overbrace{\{\tau - m - t\}}^{i\textrm{th argument}},\varnothing,\dots,\varnothing \mid T_1\cap (\tau - m - t, \tau - t),\dots,T_d\cap (\tau - m - t, \tau - t)\\
                         \Vert\;\{\tau - m - t\} \mid (\tau - m - t, \tau - t))
                     \end{multlined}                                                                                                                                                             \\
                 & = \begin{multlined}
                         j_X^+(\varnothing,\dots,\varnothing,\overbrace{\mathcal{L}_{\tau - m}\{t\}}^{i\textrm{th argument}},\varnothing,\dots,\varnothing \mid \mathcal{L}_{\tau - m}\bar{T}_1\cap \mathcal{L}_{\tau - m}(t-m, t),\dots,\mathcal{L}_{\tau - m}\bar{T}_d\cap \mathcal{L}_{\tau - m}(t-m, t)\\
                         \Vert\;\mathcal{L}_{\tau - m}\{t\} \mid \mathcal{L}_{\tau - m}(t - m, t))
                     \end{multlined}                                                                             \\
                 & = \frac{j_X^+(\mathcal{L}_{\tau - m}[\bar{T}_1\cap (t-m, t)],\dots,\mathcal{L}_{\tau - m}[\bar{T}_i\cap (t-m, t)\cup\{t\}],\dots,\mathcal{L}_{\tau - m}[\bar{T}_d\cap (t-m, t)]\;\Vert\;\mathcal{L}_{\tau - m}M[t])}{j_X^+(\mathcal{L}_{\tau - m}[\bar{T}_1\cap (t-m, t)],\dots,\mathcal{L}_{\tau - m}[\bar{T}_d\cap (t-m, t)]\;\Vert\;\mathcal{L}_{\tau - m}M[t])} \\
                 & \overset{\textrm{a.s.}}{=} \frac{f(\bar{S}^-[t] + e_i)}{f(\bar{S}^-[t])}\frac{g^+(\mathcal{L}_{\tau - m}\bar{Q}^-_1[t],\dots,\mathcal{L}_{\tau - m}[\bar{Q}^-_i\{t\}\cup\{t\}],\dots,\mathcal{L}_{\tau - m}\bar{Q}^-_d[t])}{g^+(\mathcal{L}_{\tau - m}\bar{Q}^-_1[t],\dots,\mathcal{L}_{\tau - m}\bar{Q}^-_d[t])}                                                   \\
                 & = \frac{f(\bar{S}^-[t] + e_i)}{f(\bar{S}^-[t])}\frac{g^+(\mathcal{L}_{\tau}\bar{Q}^-_1[t],\dots,\mathcal{L}_{\tau}[\bar{Q}^-_i\{t\}\cup\{t\}],\dots,\mathcal{L}_{\tau}\bar{Q}^-_d[t])}{g^+(\mathcal{L}_{\tau}\bar{Q}^-_1[t],\dots,\mathcal{L}_{\tau}\bar{Q}^-_d[t])}                                                                                                \\
                 & = \frac{f(\bar{S}^-[t] + e_i)}{f(\bar{S}^-[t])}\frac{\bar{g}(\bar{Q}_1^-[t],\dots,\bar{Q}_i^-[t]\cup\{t\},\dots,\bar{Q}_d^-[t])}{\bar{g}(\bar{Q}_1^-[t],\dots,\bar{Q}_d^-[t])}                                                                                                                                                                                      \\
                 & = \frac{f(\bar{S}^-[t] + e_i)}{f(\bar{S}^-[t])}\bar{\eta}_i^*(t)\,.
            \end{aligned}
        \end{multline*}
        Given the Gibbs sampling construction, this outcome should not be unexpected. The reverse-time process of a Gibbs sampler with systematic scan order is simply another Gibbs sampler with the scan order reversed.

        The condition for time-reversibility, as stated in \cref{eqn:reversal}, demands equality in finite-dimensional distributions. However, the Kolmogorov extension theorem states that these consistent collections of distributions are sufficient to characterize the stochastic processes \(S\) and \(\bar{S}\) entirely. We also know there exists a bijection between \(S\) and \(T\), and between \(\bar{S}\) and \(\bar{T}\). Therefore, for the processes to be time-reversible, \(T\) and \(\bar{T}\) must be equal in distribution for all \(\tau \in \mathbb{R}\). Equality in distribution for \(T\) and \(\bar{T}\) implies that, for both processes, the distribution on \(M(t)\) for some \(t\in\mathbb{R}\) must be \(j_X(\,\cdot\,\;\Vert\;M[t])\), further implying that \(g\restriction_{\mathcal{P}_{R_Y}(\mathbb{R})} = \bar{g}\restriction_{\mathcal{P}_{R_Y}(\mathbb{R})}\). But \(g\restriction_{\mathcal{P}_{R_Y}(\mathbb{R})} = \bar{g}\restriction_{\mathcal{P}_{R_Y}(\mathbb{R})}\) implies that \(\lambda(\,\cdot\, \mid \mathcal{H}_{\,\cdot\,-})\) and \(\bar{\lambda}(\,\cdot\, \mid \mathcal{H}_{\,\cdot\,-})\) are indistinguishable, which further implies equality in distribution for \(T\) and \(\bar{T}\). Therefore, time-reversibility is satisfied if and only if \(g\restriction_{\mathcal{P}_{R_Y}(\mathbb{R})} = \bar{g}\restriction_{\mathcal{P}_{R_Y}(\mathbb{R})}\) for all \(\tau \in \mathbb{R}\).}
    \end{proof}

    \Needspace{6\baselineskip}
    \subsection{Proof of \texorpdfstring{\Cref{pro:bridging}}{Proposition \ref{pro:bridging}}}\label{app:bridgingProof}

    \setcounter{proposition}{1}
    \begin{proposition}[Restated]
        As \(\varepsilon\downarrow 0\), the queue-length process \(V^{(\varepsilon)} = (V^{(\varepsilon)}(t))_{t\in\mathbb{R}_{\geq m}}\) converges in distribution to a time-inhomogeneous birth-death process \(V = (V(t))_{t\in\mathbb{R}_{\geq m}}\) with birth rates
        \begin{equation*}
            \gamma_t(y\rightarrow y + e_i) = \frac{f(y + e_i)}{f(y)}\eta_i(t)\quad(i = 1, \dots, d)
        \end{equation*}
        and death rates
        \begin{equation*}
            \gamma_t(y\rightarrow y - e_i) = y_i\eta_i(t)\quad(i = 1, \dots, d)\,.
        \end{equation*}
        Furthermore, \(V(t)\) retains the limiting distribution \(\pi\) as \(t\uparrow\infty\).
    \end{proposition}

    \begin{proof}
        {We briefly return to the point process interpretation to demonstrate that the service-time redraws retain \(\pi\) as a stationary distribution of \(V^{(\varepsilon)}\). As specified by the local Janossy density \(j_X(\,\cdot\,\;\Vert\;M[x])\) from the proof of \Cref{thm:sampler}, the distribution of the point locations given the point counts of \(X\) in \(M(x)\) is determined by \(g\). In the event that \(\eta^*\) is independent of the process history, \(g\) must factorize such that
        \begin{equation*}
            g(\{x_{1,1},\dots,x_{1,y_1}\},\dots,\{x_{d,1},\dots,x_{d,y_d}\}) = \prod_{(i,k)\in K(y)}\eta_i(x_{i,k})\,.
        \end{equation*}
        Therefore, for the point process sampler in equilibrium, the probability density function for the location of the points in \(U^{(\varepsilon)}(t)\) given \(V^{(\varepsilon)}(t) = y\) is simply \(\prod_{(i,k)\in K(y)}\varphi_i(t, t_{i,k})\). Given that the operation of redrawing the point locations from this density at some fixed time \(t\) has \(j_X(\,\cdot\,\;\Vert\;M[t])\) as a stationary distribution, and \(V^{(\varepsilon)-}(t) = V^{(\varepsilon)}(t)\) almost surely, it is clear that the addition of the service-time redraws leaves the stationary distribution of the queue-length process unchanged.

        For each \(\varepsilon\), consider decomposing the queue-length process \(V^{(\varepsilon)}\) into a difference of counting processes, such that \(V^{(\varepsilon)}(t) = N_A^{(\varepsilon)}(t) - N_D^{(\varepsilon)}(t)\), where \(N_A^{(\varepsilon)} = (N_A^{(\varepsilon)}(t))_{t\in\mathbb{R}_{\geq 0}}\) is the counting process corresponding to point arrivals and \(N_D^{(\varepsilon)} = (N_D^{(\varepsilon)}(t))_{t\in\mathbb{R}_{\geq 0}}\) is the counting process corresponding to point departures. Similarly, let the limiting queue-length process \(V\) be expressed by the decomposition \(V(t) = N_A(t) - N_D(t)\) for arrival and departure counting processes \(N_A = (N_A(t))_{t\in\mathbb{R}_{\geq 0}}\) and \(N_D = (N_D(t))_{t\in\mathbb{R}_{\geq 0}}\) respectively. These decompositions allow us to characterize the processes as \(2d\)-dimensional temporal point processes with conditional intensity functions, enabling a proof strategy very similar to that of \Cref{thm:sampler}.

        It should be highlighted that we have implicitly redefined the queue-length processes to exist on the index set \(\mathbb{R}_{\geq 0}\), instead of \(\mathbb{R}_{\geq m}\) as stated in the proposition statement. This is inconsequential for the statement, but enables the clean decompositions into counting processes. We will assume an admissible common probability distribution for the counting processes on \(M^-(m)\) across all \(\varepsilon\). The service-time redraws don't begin to occur until \(t = m\), so no departures should occur in this initial interval.

        The arrival rates for \(t \in \mathbb{R}_{\geq m}\) are straightforward. Given that a service-time redraw leaves the queue lengths unchanged, the arrival rates will always be the same as the conditional intensities of the original point process. For some \(\varepsilon \in \mathbb{R}_{>0}\), the rate at which entities join the \(i\)th queue for \(t \in \mathbb{R}_{\geq m}\) is
        \begin{equation*}
            \lambda_{A,i}^{(\varepsilon)*}(t) = \frac{f(V^{(\varepsilon)-}[t] + e_i)}{f(V^{(\varepsilon)-}[t])}\eta_i(t)\,,
        \end{equation*}
        which trivially yields the candidate arrival rates of the limiting process
        \begin{align*}
            \lambda_{A,i}^*(t) & = \lim_{\varepsilon \downarrow 0}\lambda_{A,i}^{(\varepsilon)}(t\mid \mathcal{H}_{V,t-})\quad(i = 1, \dots, d) \\
                               & = \frac{f(V^-[t] + e_i)}{f(V^-[t])}\eta_i(t)
        \end{align*}
        with respect to its history \(\mathcal{H}_{V,t-}\). It is immediate that \(\lambda_A^{(\varepsilon)}(\,\cdot\, \mid \mathcal{H}_{V,\,\cdot\,-}) \rightarrow \lambda_A(\,\cdot\, \mid \mathcal{H}_{V,\,\cdot\,-})\) in \(L^1_{\textrm{loc}}(\mathbb{R}_{\geq m}, \mu)\) almost surely.

        Now let us consider the departure rates. For \(t \in \mathbb{R}_{\geq m}\), let \(\zeta^{(\varepsilon)}(t) \in [0, \varepsilon)\) be the amount of time that has passed since the last service-time redraw. For \(\varepsilon\in(0,m]\), all entities will have experienced a service-time redraw at least once before leaving their queue. In this regime, the departure rate from the \(i\)th queue for \(t \in \mathbb{R}_{\geq m}\) is the \(\varepsilon\)-repeating hazard rate
        \begin{align*}
            \lambda_{D,i}^{(\varepsilon)*}(t) & = (N_{A,i}^{(\varepsilon)-}[t - \zeta^{(\varepsilon)}\{t\}] - N_{D,i}^{(\varepsilon)-}[t])\frac{\varphi_i(t - \zeta^{(\varepsilon)}[t], t - m)}{1 - \int_{[t - \zeta^{(\varepsilon)}(t), t]}\varphi_i(\,\cdot\, - \zeta^{(\varepsilon)}[\,\cdot\,], \,\cdot\, - m)\mathrm{d}\mu} \\
                                              & = (N_{A,i}^{(\varepsilon)-}[t - \zeta^{(\varepsilon)}\{t\}] - N_{D,i}^{(\varepsilon)-}[t])\frac{\eta_i(t)}{1 - \int_{[t - \zeta^{(\varepsilon)}(t), t]}\eta_i\mathrm{d}\mu}\,,
        \end{align*}
        giving the candidate departure rates of the limiting process
        \begin{align*}
            \lambda_{D,i}^*(t) & = \lim_{\varepsilon\downarrow 0}\lambda_{D,i}^{(\varepsilon)}(t\mid \mathcal{H}_{V,t-})\quad(i = 1, \dots, d) \\
                               & = V^-_i(t)\eta_i(t)\,.
        \end{align*}
        For the \(i\)th queue, fix \(\delta\in(0,m]\) to be a small enough number such that \(\int_{[t - \zeta^{(\varepsilon)}(t), t - \zeta^{(\varepsilon)}(t) + \varepsilon]}\eta_i\mathrm{d}\mu < 1/2\) for all \(t\) and all \(\varepsilon < \delta\). Then, for all \(\varepsilon < \delta\), the expression
        \begin{equation*}
            2V_i^-(t)\eta_i(t) > (N_{A,i}^-[t - \zeta^{(\varepsilon)}\{t\}] - N_{D,i}^-[t])\frac{\eta_i(t)}{1 - \int_{[t - \zeta^{(\varepsilon)}(t), t]}\eta_i\mathrm{d}\mu}
        \end{equation*}
        dominates the departure rates. As \(\eta_i\) is locally integrable and \(V_i^-\) is piecewise constant and almost surely finite, it follows that the dominating function \(2V_i^-(\,\cdot\,)\eta_i(\,\cdot\,)\) is also locally integrable in \(t\). By Lebesgue's dominated convergence theorem, \(\lambda_D^{(\varepsilon)}(\,\cdot\, \mid \mathcal{H}_{V,\,\cdot\,-}) \rightarrow \lambda_D(\,\cdot\, \mid \mathcal{H}_{V,\,\cdot\,-})\) in \(L^1_{\textrm{loc}}(\mathbb{R}_{\geq m}, \mu)\) almost surely.

        Corollary 4.46 of \cite{jacod2003} then guarantees that \(\lambda_A^*\) and \(\lambda_D^*\) are indeed the conditional intensity functions of the limiting process, such that \((N_A^{(\varepsilon)}, N_D^{(\varepsilon)}) \rightarrow (N_A, N_D)\) in distribution. A simple application of the continuous mapping theorem, with mapping \((u, v) \mapsto u(t) - v(t)\) applied to the counting processes, ensures \(V^{(\varepsilon)}(t) \rightarrow V(t)\) in distribution, confirming \(\pi\) as a stationary distribution of \(V\). The irreducibility property is clearly retained by \(V\), hence the limiting distribution of \(V(t)\) as \(t\uparrow\infty\) is still \(\pi\). The conditional intensities \(\lambda_A^*(t)\) and \(\lambda_D^*(t)\) evaluated at \(V^-(t) = y\) yield the birth-death process jump rates \(\gamma_t\).}
    \end{proof}

    \section{Additional simulation details}\label{app:experiments}

    For all simulations, whenever the state/queue length \(s\) was updated and the conditional intensities or transition rates needed to be recalculated, we performed the full calculation, even if a distribution-specific computational shortcut was available. This was to ensure the point process simulations adhered to \Cref{alg:point} exactly. More specifically, for the Sherrington-Kirkpatrick and stochastic neural network models, we performed the matrix multiplication \(Ws\) at every conditional intensity or transition rate update. A computational shortcut would be to, instead, simply add or subtract the \(i\)th column of \(W\) to/from the previous calculation of \(Ws\) when a point is added to or removed from the state respectively.

    To ensure the \(3\,150\) simulations were completed in a reasonable amount of time, multiple nodes of a high-performance computing cluster were used in parallel. Each simulation was assigned to a unique job, and no more than one job was run on a node at any time. Every job ran on a single machine with two Intel Xeon E5-2650 v4 processors and hyper-threading disabled.

    We have made all of the code and generated data used in our work, along instructions to reproduce the results, available at \href{https://github.com/cameronastewart/point-process-sampler}{\texttt{https://github.com/cameronastewart/point-process-sampler}}. All code is written in Python (version 3.6 and above) with minimal dependencies: NumPy is sufficient to run all simulations, while Matplotlib and a LaTeX distribution are also necessary to plot the results. Extremely small differences in \textsc{ess} values between machines are expected due to variations in floating-point arithmetic implementations. Reproducing \textsc{ess/second} exactly is not expected.

    \raggedbottom
    \bibliographystyle{apalike}
    \bibliography{references}}

\end{document}